\documentclass[11pt, a4paper]{article}
\usepackage{graphicx}
\usepackage{amscd}
\usepackage{amsmath}
\usepackage{amsfonts}
\usepackage{amssymb}
\usepackage{amsthm}
\usepackage{setspace}
\usepackage{enumerate}
\usepackage{color}
\usepackage{url}
\usepackage{hyperref}
\usepackage{bm}
\usepackage{xy}
\usepackage{stmaryrd}
\usepackage{a4wide}
\usepackage{fancyhdr}
\usepackage[dvipsnames]{xcolor}

\theoremstyle{plain}
\newtheorem{theorem}{Theorem}[section]

\newtheorem{corollary}[theorem]{Corollary}

\newtheorem{lemma}[theorem]{Lemma}

\newtheorem{proposition}[theorem]{Proposition}

\newtheorem{definition}[theorem]{Definition}

\newtheorem{assumption}[theorem]{Assumption}
\theoremstyle{remark}
\newtheorem{remark}[theorem]{Remark}

\numberwithin{equation}{section}

\newcommand{\ind}{1\!\kern-1pt \mathrm{I}}

\DeclareMathOperator*{\esssup}{\textnormal{ess\,sup}}

\newcommand{\cA}{\mathcal{A}}

\newcommand{\cC}{\mathcal{C}}

\newcommand{\cF}{\mathcal{F}}

\newcommand{\cU}{\mathcal{U}}

\newcommand{\cZ}{\mathcal{Z}}

\newcommand{\la}{\lambda}
\renewcommand{\i}{\infty}

\newcommand{\RR}{\mathbb{R}}
\newcommand{\QQ}{\mathbb{Q}}

\newcommand{\NN}{\mathbb{N}}

\newcommand{\massP}{\mathbf{P}}
\newcommand{\massQ}{\mathbf{Q}}
\newcommand{\massE}{\mathbf{E}}

\newcommand{\tT}{{0\leq t\leq T}}
\newcommand{\Econd}[2]{\massE\left[\left.#1\right|#2\right]}        
\newcommand{\Es}[1]{\massE \left[ #1 \right]}

\newcommand{\conv}{\textnormal{conv}}

\newcommand{\wt}[1]{{\widetilde{#1}}}
\newcommand{\wh}[1]{{\widehat{#1}}}
\newcommand{\sint}{\stackrel{\mbox{\tiny$\bullet$}}{}}

\usepackage{authblk}
\providecommand{\keywords}[1]{\textbf{{Keywords.}} #1}
\providecommand{\subjclass}[1]{\textbf{{MSC2010.}} #1}
\begin{document}
\title{On the existence of shadow prices for optimal investment\\ with random endowment}
\author[1]{Lingqi Gu}
\author[1, 2]{Yiqing Lin\thanks{corresponding author: yiqing.lin@polytechnique.edu}}
\author[1, 2]{Junjian Yang}
\affil[1]{\small Fakult\"at f\"ur Mathematik, Universit\"at Wien, Oskar-Morgenstern Platz 1, A-1090 Wien, Austria}
\affil[2]{\small Centre de Math\'ematiques Appliq\'ees, \'Ecole Polytechnique, F-91128 Palaiseau Cedex, France}
\date{\small \today}

\maketitle
\begin{abstract}
In this paper, we consider a num\'eraire-based utility maximization problem under constant proportional transaction costs and random endowment. 
Assuming that the agent cannot short sell assets and is endowed with a strictly positive contingent claim, a primal optimizer of this utility maximization problem exists. 
Moreover, we observe that the original market with transaction costs can be replaced by a frictionless shadow market that yields the same optimality. 
On the other hand, we present an example to show that in some case when these constraints are relaxed, the existence of shadow prices is still warranted. 
\end{abstract}

\noindent \subjclass{91B16; 91G10}

 \vspace{3mm} 
\noindent \keywords{Utility maximization, random endowment, transaction costs, no-short-selling constraints, shadow prices}

\maketitle 

\section{Introduction}
\noindent Recently, the problem of utility maximization in markets under proportional transaction costs has received many authors' attention. 
In such a market, the investor is assumed to buy securities at an ask price which is higher than the bid one that he/she receives when selling. 
The presence of proportional transaction costs enables the consideration of the portfolio optimization with models beyond semimartingales in an arbitrage-free way, which is economically meaningful (cf.~\cite{Gua06, GRS08}). 

The portfolio optimization problem under proportional transaction costs on the Merton model with the logarithm utility dates back to Magill and Constantinides \cite{MC76} and Constantinides \cite{Con86}. 
In their heuristic works, they concluded that the optimal way is to keep the current holdings in all assets in a no-trade region and to trade merely at the boundaries of such region. 
Later on, this problem has been studied extensively by many authors, among them, Taksar et al.~\cite{TKA88} first introduced the tools of singular stochastic control in the context of maximization of the rate of growth of wealth. Davis and Norman \cite{DN90} provided a rigorous formulation of the problem in \cite{MC76, Con86} and computed the location of the boundaries by solving a free boundary problem for the nonlinear partial differential equation. 
Afterwards, Shreve and Soner \cite{SS94} extended the results of \cite{DN90} by introducing viscosity solutions. We also refer to Dumas and Luciano \cite{DL91} for the maximization problem of
the asymptotic power growth rate, in which an explicit solution is constructed in a tractable case. Apart from the stochastic control approach, this utility maximization problem has also been studied with more general models by the convex duality approach. Initiated by Cvitani\'c and Karatzas \cite{CK96}, the idea of constructing an auxiliary dual problem is widely applied to various cases (e.g.~\cite{Kab99, DPT01, Bou02, BM03, CO11, CS15duality, CSY15}).

It is observed that even small transaction costs could dramatically influence the optimal choice of the investor in utility maximization (cf.~\cite{LL02}). 
Therefore, a natural question that arises is whether impacts of transaction costs on both the optimal strategy and the maximal utility can be reproduced in a frictionless market; 
 mathematically speaking, for a given utility maximization problem, we wonder whether there exists a semimartingale process lying between the bid and ask price, called {\it shadow price}, 
 such that trading in a frictionless way for this price process leads to the same optimality under transaction costs. 
The answer to this question is affirmative in finite probability spaces according to Kallsen and Muhle-Karbe \cite{KMK11}. 
However, for more general settings, this question is elusive.

With geometric Brownian motion models, Kallsen and Muhle-Karbe \cite{KMK10} considered an infinite-horizon optimal investment and consumption problem with logarithmic utility by employing tools from stochastic control and 
 constructed explicitly a shadow price by solving a free boundary problem. 
In the same framework, this idea was afterward generalized in \cite{CSZ13, HP15} for power utility to derive a shadow price. 
By the same token, Gerhold et al.~\cite{GMKS13, GGMKS14} constructed shadow prices for logarithmic and long-run power utility functions when the time horizon becomes finite.

While considering a similar problem by the convex duality method, 
the authors of \cite{CK96} showed that if the solution of a suitable dual problem could be attained by a so-called {\it consistent price system} (CPS) which is a (local) martingale, then the optimal solution could be characterized by solving a hedging problem in the ``shadow'' market associated with this CPS. 
This ``folklore'' has been clarified and sharpened very recently by Czichowsky et al.~in two aspects: 
  first, it is observed in \cite{CS15duality} that the result of Cvitani\'c and Karatzas holds true even in the framework of general c\`adl\`ag stock-price processes $S$ with only natural regularity conditions on utility functions; 
  secondly, a sufficient (but not necessary) condition for ensuring the local martingale property of the optimal dual processes is found in \cite{CSY15}, that is, $S$ is continuous and satisfies ``no unbounded profit with bounded risk" $(NUPBR)$. 
Subsequently, this condition is replaced by a weaker condition of $(TWC)$ of ``two way crossing'' in \cite{CPSY16}.
It is worth mentioning that the notion of shadow price is generalized to a sandwiched form in \cite{CS15duality}, which is adapted to the market driven by a stock-price process with c\`adl\`ag paths. 
This idea is afterward  adopted by Bayraktar and Yu \cite{BY15} in order to study a similar problem but with random endowment: they constructed shadow prices in the sandwiched sense as \cite{CS15duality}, whenever the duality result holds and a sufficient condition on the dual optimizer is assumed.
When the utility function is defined on the whole real line, Czichowsky and Schachermayer \cite{CS15portfolio} proved that the dual optimizer defines a shadow price as long as the primal optimizer is attainable. 
In particular, this can be ensured if $S$ is an exponential fractional Brownian motion. 
In addition, we refer to \cite{LY16} for the generalization of the results in \cite{CS15portfolio} with random endowment.

Instead of studying the dual optimizer, Loewenstein \cite{Loe00} constructed shadow markets directly from the derivatives of dynamic primal value functions under no-short-selling constraints. 
This argument is based on the existence of the constrained primal solution, which was a hypothesis assumed by \cite{Loe00} but has been affirmed by Benedetti et al.~\cite{BCKMK13} 
  when they applied Loewenstein's approach to a similar problem with Kabanov's multi-currency model. 

In this paper, we consider on the one hand a num\'eraire-based utility maximization problem with constant proportional transaction costs and random endowment under no-short-selling constraints. In contrast to \cite{BY15}, our method is to study straight away the primal problem without the formulation of  the dual problem. First, we are inspired by \cite[Section 3.3]{Sch04LN} and prove the existence of constrained primal solutions. 
Second, under the assumption that the agent is endowed with a positive random endowment rather than a deterministic initial wealth, we follow the lines of \cite{Loe00, BCKMK13} to construct a shadow price directly from the primal solution. Comparing with \cite{BY15}, our result aims to prove the existence of a classical shadow price process instead of the sandwiched one and moreover the regularity conditions assumed in \cite{BY15} on the random endowment is removed. 
On the other hand, we discuss the existence of shadow prices when the constraints are violated and the random endowment is allowed to be negative. 
We provide an example in the Black-Scholes framework  with a constructive random endowment. In this example, shadow prices exist
and can be explicitly defined.

This paper is organized as follows. 
In Section 2, we introduce notations and formulate the utility maximization problem under proportional transaction costs with random endowment.
Section 3 presents our main result, i.e., the existence of shadow prices under the no-short-selling constraint. 
Next, we provide in Section 4 an example which falls out of the framework studied in Section 3, however shadow prices can be explicitly defined. 
%
%
%
%

\section{Formulation of the problem}

In this section, we shall briefly introduce the basic setting of the utility maximization problem in markets with random endowment and transaction costs, as well as the definition of the shadow price in the classical sense. 
The reader, who has more interests in the details of these topics, is referred to \cite{KMK10,KMK11,BCKMK13,CMKS14,CSY15}.

We consider a num\'eraire-based model: the financial market consists of two assets, one bond and one stock, where the price of the bond $B$ is constant and normalized to $B\equiv 1$. 
We denote by $S=(S_t)_{0\leq t\leq T}$ the price process of the stock, 
 which is based on a filtered probability space $(\Omega, \cF, (\cF_t)_{0\leq t\leq T}, \massP)$ satisfying the usual hypotheses of right continuity and saturatedness, 
 where $\cF_0$ is assumed to be trivial.  
Here, $T$ is a finite time horizon. 
In the sequel, we denote $L^0(\Omega, \mathcal{F}, \massP)$ by $L^0$ and $L^1(\Omega, \mathcal{F}, \massP)$ by $L^1$. 
Throughout the paper we make the following assumptions:

\begin{assumption} \label{Sassumption}
 The price process $S=(S_t)_{0\leq t\leq T}$ is adapted to $(\cF_t)_{0\leq t\leq T}$, with c\`adl\`ag and strictly positive paths. Additionally, $\mathcal{F}_{T-}=\mathcal{F}_T$ and $S_{T-}=S_T$.
\end{assumption}

Similarly to \cite{CSY15}, we introduce constant proportional transaction costs $0<\lambda<1$ for the trading of the stock, which models the width of the bid-ask spread.
The agent has to pay a higher ask price $S_t$ to buy stock shares but only receives a lower bid price $(1-\lambda)S_t$ when selling them. 
 
As the counterpart of martingale measures in the frictionless case, consistent price systems (CPSs) play a very important role in the framework with transaction costs (compare, e.g., \cite{KS09, Sch14super}). 
In the present paper, to establish the utility maximization problem, we adopt an extended notion -- $\lambda$-supermartingale-CPSs, similarly defined as in \cite{BCKMK13}.

\begin{definition} \label{CPS}
 Fix $\lambda>0$ and the price process $S=(S_t)_{0\le t\le T}$ satisfying Assumption \ref{Sassumption}.
 A $\lambda$-supermartingale-CPS is a couple of two positive processes $Z=(Z^0_t,Z^1_t)_{0\le t\le T}$ consisting of two supermartingales $Z^0$ and $Z^1$, 
  such that
  \begin{equation}\label{J10}
   {S}^Z_t:=\frac{Z^1_t}{Z^0_t} \in [(1-\lambda)S_t, S_t],\quad \mbox{a.s.},
  \end{equation}
  for all $0\le t\le T$.
  
 The set of all $\lambda$-supermartingale-CPSs is denoted by $\cZ^\lambda_{sup}$.
\end{definition}

We introduce now the following assumption on the existence of a supermartingale-CPS:

\begin{assumption}\label{Scps}
 For some $0<{\lambda'}<\lambda$, we have that $\mathcal{Z}^{\lambda'}_{sup}\neq\emptyset$.
\end{assumption}
 \begin{remark}
 We remark that this condition corresponds to the ``no unbounded profit with bounded risk'' with no-short-selling constraint on portfolios in the frictionless theory,  
 see \cite{KK07}.
 \end{remark}
%
\begin{remark}
Instead of all $0<{\lambda'}<\lambda$, we only need that the condition above holds for some $0<{\lambda'}<\lambda$, similar to Lemma 3.1 in \cite{Sch14super}. This condition is sufficient to show the convex compactness and the $L^0$-boundedness of the set of admissible terminal wealth.
\end{remark}

In this market under transaction costs, the following definition of self-financing trading strategies is commonly adopted, e.g., in \cite{Sch14admissible, Sch14super}.  

\begin{definition}
 Fix $\lambda\in(0, 1)$. 
 A self-financing trading strategy under transaction costs $\lambda$ is a predictable $\RR^2$-valued finite variation process 
  $\varphi = (\varphi_t^0, \varphi_t^1)_{0\leq t\leq T}$ 
  such that 
    $$ \int_s^td\varphi^{0,\uparrow}_u = \int_s^t(1-\lambda)S_ud\varphi_u^{1,\downarrow}, \quad \int_s^td\varphi^{0,\downarrow}_u = \int_s^tS_ud\varphi_u^{1,\uparrow}, $$
  for all $0\leq s\leq t\leq T$, where $\varphi^{\uparrow}$ and $\varphi^{\downarrow}$ denote the Jordan-Hahn decomposition of $\varphi$. 
\end{definition}

The processes $\varphi_t^0$ and $\varphi_t^1$ describe the amount of bond and the number of stock shares held at time $t\in [0,T]$. 

Assume that as well as being able to trade on the financial market, the agent is endowed with a positive random endowment at the terminal time $T$, which is described by an $\mathcal{F}_T$-measurable random variable.
\begin{assumption} \label{re}
 The endowment $\overline{e}_T\in \mathcal{F}_T$ is a strictly positive and finite-valued random variable, which can be decomposed into a deterministic part and a random part, i.e., 
  $\overline{e}_T=x+e_T$, where $x>0$ and $e_T\geq 0$, a.s. We assume that $x$ is the initial wealth of the agent and $e_T$ is endowed at time $T$.  
\end{assumption} 
\begin{remark}
In the financial market, $e_T$ can be explained as a positive contingent claim, e.g., an option contract.
\end{remark}

In this paper, we shall consider a utility maximization problem similar to the one in Benedetti et al.~\cite{BCKMK13}, where the agent is facing the no-short-selling constraint, 
  which forces him to keep both the amount of bond and the number of stock shares positive. 
In other words, with the initial wealth $x$, the agent is only allowed to trade with the admissible strategies defined as follows:

\begin{definition}
 Under transaction costs $\lambda\in(0, 1)$, a self-financing strategy $\varphi=(\varphi^0,\varphi^1)$ with 
   $(\varphi_{0}^0, \varphi_{0}^1)=(x, 0)$ and $(\varphi_{T}^0, \varphi_{T}^1)=(\varphi_{T}^0, 0)$ is called admissible, 
  if for each $t\in[0,T]$ we have that 
   $$\varphi_t^0\geq 0\, \mbox{ and }\, \varphi_t^1\geq 0,\quad \mbox{a.s}. $$
 We denote by $\mathcal{A}^\lambda(x)$ the collection of all such strategies.
 Moreover, we define
   $$ \mathcal{C}^\lambda(x):=\left\{g\in L^0_+ \,\Big|\, g\leq\varphi^0_T,\,\mbox{ for some } \,(\varphi^0, \varphi^1)\in\mathcal{A}^\lambda(x)\right\}. $$
\end{definition}

We suppose that the agent's preferences over terminal wealth are modeled by a utility function $U:(0,\infty)\to\RR$, 
  which is strictly increasing, strictly concave, continuously differentiable. 
Furthermore, the function $U$ satisfies the Inada conditions, i.e., 
  $$ U'(0) := \lim_{x\to 0}U'(x) = \infty \quad \textnormal{ and } \quad U'(\infty) := \lim_{x\to \infty}U'(x) = 0. $$
and the following condition of reasonable asymptotic elasticity (RAE).

\begin{assumption} \label{U(x)assumption}
 The utility function $U$ satisfies the reasonable asymptotic elasticity, i.e.,
   $$ AE(U):= \limsup_{x\to\infty}\frac{xU'(x)}{{U}(x)}< 1,$$
\end{assumption}
Denote $V: \RR_+\to\RR$ the convex conjugate function of $U$ defined by 
$$ V(y):=\sup_{x>0}\{U(x)-xy\}, \quad y>0,$$
which is strictly decreasing, strictly convex and continuously differentiable and satisfies 
$$V'(0)=-\infty, \quad  V'(\infty)=0,  \quad  V(0)=U(\infty),  \quad   V(\infty)=U(0). $$

For financial interpretation and more results about the previous assumption, we refer to \cite{KS99} and \cite{KS03}.

Without loss of generality, we may assume $U(\infty)>0$ to simplify the analysis and we define $U(x)=-\infty$ whenever $x\leq 0$.

Then, the problem for the agent is to maximize expected utility at terminal time $T$ from his bond account derived from trading and the random endowment, i.e., 
\begin{equation} \label{pp}
  u(x;e_T):=\sup_{g\in \cC^{\lambda}(x)}\massE [U(g+e_T)]. 
\end{equation}

For $Z\in \cZ^\lambda_{sup}$, define $S^Z:=\frac{Z^1}{Z^0}$. 
By the definition, $S^Z$ is a positive semimartingale taking values in $[(1-\lambda)S, S]$. 
Then, we can construct a frictionless market consisting of one bond with zero interest rate and an underlying asset, whose price process is $S^Z$. 
Adapting the previous setting under transaction costs, we adopt the following notion of self-financing trading strategies. 

\begin{definition}
 In the frictionless market associated with $S^Z$, an $\mathbb{R}^2$-valued predictable process $\widetilde{\varphi}:=(\widetilde{\varphi}_t^0, \widetilde{\varphi}_t^1)$ 
   starting from $(x,0)$ is a self-financing trading strategy, 
   if $\widetilde{\varphi}^1$ is $S^Z$-integrable and 
    $$ \widetilde{\varphi}^0_t+\widetilde{\varphi}^1_tS^Z_t = x+\int^t_0\widetilde{\varphi}_u^1dS^Z_u,\quad 0\leq t\leq T. $$
 Here, $\widetilde{\varphi}_t^0$ and $\widetilde{\varphi}_t^1$ describe the amount of bond and the number of stock shares held at time $t\in[0,T]$.
\end{definition}

We shall formulate a utility maximization problem for the frictionless model with $S^Z$. 
In accordance with \eqref{pp}, we always assume that neither asset can be shorted, so that the maximization problem is established over all admissible strategies defined as follows.

\begin{definition}
 Let $S^Z:=\frac{Z^1}{Z^0}$, for some $Z\in \cZ^\lambda_{sup}$. 
 A self-financing strategy $\widetilde{\varphi}$ is admissible in the market driven by $S^Z$, if we have
  $$ \widetilde{\varphi}_t^0\geq 0 \,\mbox{ and }\, \widetilde{\varphi}_t^1\geq 0,\quad a.s., $$
  for all $0\leq t\leq T$.
  
 We denote by $\mathcal{A}^Z(x)$ the collection of all such admissible trading strategies starting from $(x, 0)$.  
 Moreover, we define
  $$ \mathcal{C}^Z(x):=\left\{\widetilde{g}\in L^0_+ \,\big|\, \wt g\leq \widetilde{\varphi}^0_T+\widetilde{\varphi}^1_TS^Z_T,\,\mbox{ for some }\,(\widetilde{\varphi}^0,\widetilde{\varphi}^1)\in\mathcal{A}^Z(x)\right\}. $$
\end{definition}

\begin{lemma}
 The payoff of a trading strategy in the market $S$ with transaction costs can be dominated by some outcomes from trading in the potentially more favorable frictionless markets driven by $S^Z$, $Z\in\cZ^\lambda_{sup}$.
 Namely, fix $Z\in\cZ^\lambda_{sup}$ and let $(\varphi^0,\varphi^1)\in\cA^\lambda(x)$ be arbitrary,
   then there exists a $(\wt\varphi^0, \wt\varphi^1)\in\cA^Z(x)$ such that 
    \begin{equation} \label{no-short}
      \wt\varphi^0_t \geq \varphi^0_t \,\mbox{ and }\, \wt\varphi^1_t \geq \varphi^1_t,  \quad a.s.,
    \end{equation}
   for all $\tT$.
\end{lemma}

\begin{proof}
 For any $(\varphi^0, \varphi^1)\in\cA^\lambda(x)$, since $(\varphi^0,\varphi^1)$ is a $\lambda$-self-financing trading strategy, we have that
  \begin{align*}
   \varphi^0_t+ \varphi^1_t S_t^Z &= x + \int_0^t d \varphi^0_u + \int_0^t \varphi^1_{u}dS^Z_u +\int_0^t S^Z_u d\varphi^1_{u} \\
                                  &= x + \int_0^t (d \varphi^0_u +S^Z_u d\varphi^1_{u})+ \int_0^t \varphi^1_{u}d S^Z_u 
                                \leq x +  \int_0^t \varphi^1_{u}d S^Z_u. 
  \end{align*}
  We remark that we adopt throughout this paper the notion of stochastic integrals in \cite[Section 7]{CS16strong}. Then, define a self-financing trading strategy in the frictionless market associated with $S^Z$ by
  \begin{align*}
    \left\{\begin{array}{l} \displaystyle
            \wt\varphi_t^0 := x + \int_0^t \varphi^1_{u}d S^Z_u - \varphi^1_{t} S^Z_t \geq \varphi_t^0, \\
            \wt\varphi_t^1 := \varphi_t^1, 
           \end{array} 
    \right. 
   \end{align*} 
   which satisfies \eqref{no-short}.
 \end{proof}
 
Obviously, $\cC^\lambda(x)\subseteq\cC^Z(x)$, for any $Z\in\cZ^\lambda_{sup}$. 
Therefore, letting 
   $$ u^Z(x;e_T):=\sup_{\widetilde{g}\in \cC^Z(x)}\massE [U(\widetilde{g}+e_T)], $$ 
  it follows that 
   $$ u(x;e_T)\leq \inf_{Z\in \cZ^\lambda_{sup}} u^Z(x;e_T), $$
  which means each frictionless market with $S^Z$ affords better, at least not worse, investment opportunity than the  market with transaction costs. 
An interesting question is whether there exists a least favorable $\widehat{Z}\in \cZ^\lambda_{sup}$, such that the {gap is closed}, i.e., the inequality becomes equality. 
If so, the corresponding price process $S^{\widehat{Z}}:=\frac{\widehat{Z}^1}{\widehat{Z}^0}$ is called shadow price. 
Below is the definition of the shadow price similar to \cite[Definition 3.9]{BCKMK13}.

\begin{definition}
 Fix the initial value $x$ and the terminal random endowment $e_T$. 
 We assume that short selling of either asset is not allowed. 
 Then, the process $S^{\widehat{Z}}$ associated with some $\widehat{Z}:=\widehat{Z}(x, e_T)\in\cZ^\lambda_{sup}$ is called a shadow price process, if 
  $$ \sup_{g\in \cC^\lambda(x)}\massE[U(g+e_T)] = \sup_{\widetilde{g}\in\cC^{\widehat{Z}}(x)}\massE[U(\widetilde{g}+e_T)]. $$ 
\end{definition}

\section{Solvability of the problem and existence of shadow prices} 
%
In this section, we shall present our main result, that is, the solvability of \eqref{pp} and the existence of shadow prices. 
\subsection{Main theorems}


The existence of shadow prices for the utility maximization problem with neither the no-short-selling constraint nor random endowment has been studied in \cite{KMK11,CMKS14, CS15duality, CSY15, CPSY16} by duality methods. 
By contrast, we shall solve \eqref{pp} directly by following the line of \cite{BCKMK13}. 

First of all, we display a superreplication theorem as an analogue of \cite[Lemma 4.1]{BCKMK13}: 

\begin{lemma} \label{suprep}
 For any $Z\in \mathcal{Z}^\lambda_{sup}$, the process $Z^0\varphi^0+ Z^1\varphi^1$ is a positive supermartingale, 
   for any $(\varphi^0, \varphi^1)\in\cA^\lambda(x)$. 
\end{lemma}

\begin{proof}
 As $(\varphi^0,\varphi^1)$ is of finite variation and $(Z^0,Z^1)$ is a supermartingale, we obtain that
  \begin{align*}
    Z^0_t\varphi^0_t+Z^1_t\varphi^1_t &= (Z^0_0\varphi^0_0+Z^1_0\varphi^1_0)+\int_0^t(\varphi^0_{u}d Z_u^0 + \varphi^1_{u}d Z_u^1)+\int_0^t(Z_u^0d\varphi^0_{u}+Z_u^1d\varphi^1_{u}) \\
                                      &= x + \int_0^t(\varphi^0_{u}d Z_u^0 + \varphi^1_{u}d Z_u^1)+\int_0^t(Z_u^0d\varphi^0_{u}+Z_u^1d\varphi^1_{u}).
  \end{align*}
 The first integral defines a supermartingale due to the positivity of $\varphi^0$ and $\varphi^1$. 
 The second integral defines a decreasing process by the fact that $(\varphi^0,\varphi^1)$ is $\lambda$-self-financing and that $\frac{Z^1_u}{Z^0_u}$ takes values in $[(1-\lambda)S_u,S_u]$. 
 Therefore, the process $Z^0\varphi^0+ Z^1\varphi^1$ is a positive supermartingale. 
\end{proof}

\begin{remark}
Comparing with \cite[Theorem 1.4]{Sch14super}, we require less on the underlying asset price $S$ for the superreplication theorem, since we are working with a smaller set of trading strategies.
\end{remark}

Furthermore, we have some properties of the convex sets $\cA^\lambda(x)$ and $\cC^\lambda(x)$ as follows.

\begin{lemma}\label{TVb}
 Under Assumption \ref{Sassumption} and \ref{Scps}, the total variation ${\rm Var}(\varphi^0)$ and ${\rm Var}(\varphi^1)$ remain bounded in $L^0$, when $\varphi$ runs through $\cA^\lambda(x)$.
\end{lemma}

\begin{proof}
 Write $\varphi^0=\varphi^{0,\uparrow}-\varphi^{0,\downarrow}$ and $\varphi^1=\varphi^{1,\uparrow}-\varphi^{1,\downarrow}$ as the canonical differences of increasing processes. 
 Then, we could define a strategy $\widetilde{\varphi}\in\mathcal{A}^{\lambda'}(x)$ by 
   $$ \widetilde{\varphi}_t:=\left(\varphi^0_t+\frac{\lambda-\lambda'}{1-\lambda}\varphi^{0, \uparrow}_t, \varphi^1_t\right), \quad 0\leq t\leq T, $$ 
  and prove by Lemma \ref{suprep} that for $Z\in \cZ^{\lambda'}_{sup}$, 
   $$ \frac{\lambda-\lambda'}{1-\lambda}\massE\left[Z^0_T\varphi^{0, \uparrow}_T\right]\leq \massE[Z^0_T\varphi^0_T+Z^1_T\varphi^1_T]+\frac{\lambda-\lambda'}{1-\lambda}\massE\left[Z^0_T\varphi^{0, \uparrow}_T\right]\leq x. $$
   Inspired by the proof of \cite[Lemma 3.2]{CS06}, we construct a probability measure $\mathbf{Q}$ by defining 
   $$
   \frac{d\mathbf Q}{d\mathbf P}=\frac{g}{\mathbf{E}[g]},
   $$
   where $g:=\frac{\inf_{t\in [0, T]}Z^0_t}{\mathbf{E}[Z^0_T]}>0$. Then, we have 
   $$ \frac{\lambda-\lambda'}{1-\lambda}\massE^\mathbf{Q}\left[\varphi^{0, \uparrow}_T\right] \leq \frac{x}{\mathbf{E}[g]\mathbf{E}[Z_T^0]}. $$
 The reminder of the proof is identical with the one of \cite[Lemma 3.1]{Sch14super}.
\end{proof}
Then, we state the following lemma without proof and refer the reader to \cite[Theorem 3.4]{Sch14super}.
\begin{lemma} \label{l1}
 Under Assumptions \ref{Sassumption} and \ref{Scps}, the set $\cC^\lambda(x)$ is convex closed and bounded in $L^0_+$.
\end{lemma}

In what follows, we shall establish the existence and uniqueness result for the primal solution of \eqref{pp}.
In the frictionless case, a similar result without random endowment has been proved in \cite{Sch04LN} (compare also \cite{Gua02} and \cite{BCKMK13}). The author of \cite{Sch04LN} applied a technical lemma (\cite[Lemma 3.16]{Sch04LN}) to conduct a proof by contradiction and show that the limit of a maximizing sequence indeed solves the utility maximization problem.  We find that the assumption $\{f_n\}_{n\in\mathbb{N}}\geq 0$ in \cite[Lemma 3.16]{Sch04LN} is not essentially needed for proceeding the argument. Thus, 
 we reorganize a lemma in the Appendix and for the sake of completeness, we also give the proof of the following theorem. 
%
%
\begin{theorem} \label{MT1}
 Let Assumptions \ref{Sassumption}, \ref{Scps}, \ref{re} and \ref{U(x)assumption} hold. 
 Assume moreover that $$ u(x; e_T)<\infty. $$ 
 Then, the utility maximization problem \eqref{pp} admits a unique solution $\widehat{g}\in\cC^\lambda(x).$
\end{theorem}
%
%
\begin{proof}
 The uniqueness is trivial due to the strict concavity of $U$. Thus, we only have to show the existence. 
 
 \vspace{2mm}
 
 (i) Since $u(x; e_T)<\infty$, we could find a maximizing sequence for \eqref{pp}, i.e.,
 $$ u(x; e_T)=\lim_{n\to\infty}\massE[U(g_n+e_T)].$$
     By passing to a sequence of convex combinations $\conv(g_n, g_{n+1}, \ldots)$, still denoted by $g_n$, and applying Lemma \ref{l1} as well as the Koml\'os-type theorem (e.g.~\cite[Lemma A1.1a]{DS94}), 
       we may suppose that $g_n$ converges a.s.~to $\widehat{g}\in\mathcal{C}^\lambda(x)$.
   
 \vspace{2mm}  
   
 (ii)  We claim that $(U(\widehat{g}+e_T))^+$ is integrable and thus $\massE[U(\widehat{g}+e_T)]$ exists.
  Without loss of generality we assume that $U(1)=0$. For any $g \in \cC^{\lambda}(x)$,  $g+1 \in \cC^{\lambda}(x+1)$.  
  It is easy to verify that $u$ is still a concave function in $x$ and thus $u(x; e_T)<\infty$ implies $u(x+1; e_T)<\infty$. 
 We suppose  for the sake of contradiction that $(U(\widehat{g}+e_T))^+$ is not integrable, then
 \begin{align*} 
  \massE[(U(\widehat{g}+e_T))^+] <\massE[(U(\widehat{g}+1 +e_T))^+] = \massE[U(\widehat{g}+1 +e_T)] \leq u(x+1) <\infty,
 \end{align*}
which is contradiction.  
  

 \vspace{2mm}   
    
 (iii) We now prove that $\widehat{g}$ is the primal optimizer. 
       If not, there exists an $\alpha\in (0, \infty]$ such that 
         $$ \alpha=u(x; e_T)-\massE[U(\widehat{g}+e_T)]. $$
       For each $n$, denote $f_n=U(g_n+e_T)$ and denote $f_0=U(\widehat{g}+e_T)$.
       Fixing $\varepsilon>0$, there exists an ${m'}\in \mathbb{{N}}$, such that for each $n\geq {m'}$, 
\begin{align}\label{exp1} u(x; e_T)-\massE[f_{n}]\leq \varepsilon. \end{align}
       Since $AE(U) < 1$, by \cite[Lemma 6.3]{KS99}, there exists some $\gamma > 1$, such that $U(\frac{x}{2})>\frac{\gamma}{2} U(x)$, for all $x \geq x_0>0$. 
      Note that for each $n\in \mathbb{N}$ and any ${M} >0$, 
    \begin{align*} 
      \massP\big[f_n\geq {M}\big] \leq \massP\left[|f_n-f_0|\geq \frac{{M}}{2}\right]+\massP\left[f^+_0\geq \frac{{M}}{2}\right].
    \end{align*}
        Thus, for any $\delta>0$, we can choose sufficiently large $M>0$ with $U^{-1}(M)\geq 2x_0$ and find a $m_0\geq m'$ such that for any $n\geq m_0$, 
        $\massP[f_n\geq M]\leq\delta.$
        Due to the integrability of $f_{m'}$, for properly chosed $\delta$, 
        $ \massE\big[|f_{m'}|{\bf 1}_{\{f_n\geq M\}}\big]\leq {\varepsilon} $ holds for any $n\geq m_0$.
       From Lemma \ref{lemnew}, we fix  $m>m_0$ such that 
  \begin{align}\label{exp2}  \massE\big[f_m{\bf 1}_{\{f_m\geq M\}}\big]\geq \alpha-{\varepsilon} \quad\mbox{ and }\quad \massE\big[|f_{{m'}}|{\bf 1}_{\{f_m\geq M\}}\big]\leq {\varepsilon}. \end{align}
 Then, 
        \begin{align*}
          \massE\left[U\left(\frac{g_{m'}+g_m}{2}+e_T\right)\right] &= \massE\left[U\left(\frac{g_{m'}+g_m}{2}+e_T\right){\bf 1}_{\{f_m\geq M\}}\right]\\
                                                             &\hspace{5mm}+ \massE\left[U\left(\frac{g_{m'}+g_m}{2}+e_T\right){\bf 1}_{\{f_m< M\}}\right].
        \end{align*}
       Furthermore, due to the positivity of $g_{m'}$, $g_m$ and $e_T$, we have
        \begin{align*}
          \massE\left[U\left(\frac{g_{m'}+g_m}{2}+e_T\right){\bf 1}_{\{f_m\geq M\}}\right]
            \geq  \frac{\gamma}{2}\massE\left[U\left(g_{m'}+g_m+2e_T\right){\bf 1}_{\{f_m\geq M\}}\right]
            \geq  \frac{\gamma}{2}\massE\left[f_m{\bf 1}_{\{f_m\geq M\}}\right]
        \end{align*}
        and 
        \begin{align*}
          \massE\left[U\left(\frac{g_{m'}+g_m}{2}+e_T\right){\bf 1}_{\{f_m< M\}}\right]
            \geq \frac{1}{2}\massE\left[f_{m'}{\bf 1}_{\{f_m< M\}}\right]+\frac{1}{2}\massE\left[f_m{\bf 1}_{\{f_m< M\}}\right].
        \end{align*}
       Therefore, we can deduce from (\ref{exp1}) and (\ref{exp2}) that 
       \begin{align*}
          \massE\left[U\left(\frac{g_{m'}+g_m}{2}+e_T\right)\right]&\geq \frac{1}{2}\massE\left[f_{m'}{\bf 1}_{\{f_m< M\}}\right]+\frac{1}{2}\massE\left[f_m\right]+\frac{\gamma-1}{2}\massE\left[f_m{\bf 1}_{\{f_m\geq M\}}\right]\\
          &\geq u(x; e_T)+\frac{(\gamma-1)\alpha}{2}-\frac{\gamma+2}{2}\varepsilon.
         \end{align*}
       Letting $\varepsilon\rightarrow 0$, we have 
         $$ \massE\left[U\left(\frac{g_{m'}+g_m}{2}+e_T\right)\right]> u(x; e_T), $$
         which is a contradiction to the maximality of $u(x; e_T)$.
\end{proof}

Now we turn to consider the frictionless market associated with $S^Z$, for $Z\in\cZ^\lambda_{sup}$. 
Similarly to Lemma \ref{suprep}, we have the supermartingale property of $Z^0\widetilde{\varphi}^0+ Z^1\widetilde{\varphi}^1$ for $(\widetilde{\varphi}^0, \widetilde{\varphi}^1)\in \cA^Z(x).$
The subsequent lemma has been reviewed in \cite[Lemma 4.1]{BCKMK13}. 
However, for the convenience of the reader, we prove it in the num\'eraire-based case. 

\begin{lemma} \label{suprepz}
 Fix $Z\in\cZ^\lambda_{sup}$. 
 The process $Z^0\widetilde{\varphi}^0+ Z^1\widetilde{\varphi}^1$ is a positive supermartingale, for any $(\widetilde{\varphi}^0, \widetilde{\varphi}^1)\in \cA^Z(x).$
\end{lemma}

\begin{proof}
 Note that $Z^0\widetilde{\varphi}^0+ Z^1\widetilde{\varphi}^1 = Z^0\big(\wt\varphi^0+\wt\varphi^1S^Z\big) = Z^0\big(x+\wt\varphi^1\sint S^Z\big)$ 
   by the frictionless self-financing condition. 
 Using It\^{o}'s formula and \cite[Proposition A.1]{GK03}, we obtain that 
   \begin{align*}
     Z^0_t\widetilde{\varphi}^0_t+ Z^1_t\widetilde{\varphi}^1_t
      &= xZ^0_0 + \big((\wt\varphi^0_-+\wt\varphi^1_-S_-^Z)\sint Z^0\big)_t + \big(Z_-^0\sint(\wt\varphi^1\sint S^Z)\big)_t + \big[\wt\varphi^1\sint S^Z, Z^0\big]_t \\
      &= xZ^0_0 + \big((\wt\varphi^0_-+\wt\varphi^1_-S_-^Z)\sint Z^0\big)_t + \big(\wt\varphi^1\sint(Z_-^0\sint S^Z)\big)_t 
                + \big(\wt\varphi^1\sint\big[S^Z, Z^0\big]\big)_t. 
   \end{align*}  
 It follows from the frictionless self-financing condition again and \cite[I.4.36]{JS03} that 
   $$ \Delta(\wt\varphi^0+\wt\varphi^1S^Z) = \Delta(\wt\varphi^1\sint S^Z) = \wt\varphi^1\Delta S^Z, $$
   therefore, $\wt\varphi^0_-+\wt\varphi^1_-S^Z_- = \wt\varphi^0 + \wt\varphi^1S^Z_-$.
 By \cite[I.4.37, Definition I.4.45]{JS03}, we obtain 
  \begin{align*}
    Z^0_t\widetilde{\varphi}^0_t+ Z^1_t\widetilde{\varphi}^1_t
      & = xZ^0_0 + \big(\wt\varphi^0\sint Z^0\big)_t + \big(\wt\varphi^1\sint\big(S^Z_-\sint Z^0+Z_-^0\sint S^Z+[S^Z, Z^0]\big)\big)_t \\
      & = xZ^0_0 + \big(\wt\varphi^0\sint Z^0\big)_t + \big(\wt\varphi^1\sint(S^ZZ^0)\big)_t  \\
      & = xZ^0_0 + \big(\wt\varphi^0\sint Z^0\big)_t + \big(\wt\varphi^1\sint Z^1\big)_t, 
  \end{align*}
  which is a positive local supermartingale and hence a supermartingale. 
\end{proof}

Then, it is easy to deduce that for each $Z\in\cZ^\lambda_{sup}$ and $\widetilde{g}\in\cC^Z(x)$, 
 \begin{align} \label{onedir}
   \massE[U(\widetilde{g}+e_T)]\leq \massE[V(Z^0_T)+Z^0_T(\widetilde{g}+e_T)] \leq \massE[V(Z^0_T)]+\massE[Z^0_Te_T]+Z^0_0x.
 \end{align}
Therefore, to prove the existence of shadow prices, it suffices to have the following lemma, whose proof is postponed to the next subsection.

\begin{lemma} \label{condition}
 Let Assumptions \ref{Sassumption}, \ref{Scps}, \ref{re} and \ref{U(x)assumption} hold. 
 There exists a $\widehat{Z}\in \cZ^\lambda_{sup}$, such that 
  \begin{enumerate}[$(i)$]
    \item $\widehat{Z}^0_T=U'(\widehat{g}+e_T)$;
    \item $\massE[\widehat Z^0_T\widehat{g}]=\widehat{Z}^0_0x.$
  \end{enumerate}
\end{lemma}

\begin{theorem}
 The $\lambda$-supermartingale-CPS $\widehat{Z}\in\cZ^\lambda_{sup}$ satisfying Lemma \ref{condition} (i)-(ii) defines a shadow price 
   $S^{\widehat{Z}}:=\frac{\widehat{Z}^1}{\widehat{Z}^0}$.
\end{theorem}

\begin{proof}
 Consider the frictionless market associated with $S^{\widehat{Z}}$.
 By Lemma \ref{condition}, we have
  \begin{align*}
    u^{\widehat{Z}}(x;e_T)\geq u(x; e_T) &= \massE[U(\widehat{g}+e_T)]=\massE[V(\widehat{Z}^0_T)+\widehat{Z}^0_T(\widehat{g}+e_T)] \\
                                    &= \massE[V(\widehat{Z}^0_T)]+\massE[\widehat{Z}^0_Te_T]+\widehat{Z}^0_0x\geq u^{\widehat{Z}}(x;e_T),
  \end{align*}
  where the last inequality follows from \eqref{onedir}.
 The inequality above implies $u^{\widehat{Z}}(x;e_T)= u(x;e_T)$, which proves that $S^{\widehat{Z}}$ is a shadow price for the problem \eqref{pp}.
\end{proof}

\begin{remark}
 By the strict concavity of $U$, $\widehat{g}$ is the unique solution in $\mathcal{C}^{\widehat{Z}}(x)$ for the frictionless problem $u^{\widehat{Z}}(x;e_T)$. 
 Moreover, the trading strategy $\widehat{\varphi}$, that attains the maximum in the market with transaction costs, does the same in the frictionless one associated with the shadow price $S^{\widehat{Z}}$. 
 Therefore, the optimal trading strategy $(\widehat{\varphi}^0,\widehat{\varphi}^1)$ for $S$ under transaction costs $\lambda$ satisfies 
  \begin{align*}
    & \{d\widehat{\varphi}^1_t>0\}\subseteq \{S_t^{\widehat{Z}}=S_t\}, \\
    & \{d\widehat{\varphi}^1_t<0\}\subseteq \{S^{\widehat{Z}}_t=(1-\lambda)S_t\},
  \end{align*}
  for all $0\leq t\leq T$.
\end{remark}

\begin{remark}
  In our case, shadow prices are determined not only by the random endowment but also by its decomposition (see Assumption \ref{re}).  
  The decomposition of the random endowment together with the no-short-selling constraints can be explained as the agent's trading rule created by his/her controller. 
  Precisely, if the random endowment $\widetilde{e}_T$ that the agent will eventually receive is decomposed into $x+e_T$ by his/her controller, then it means that the agent is allowed to spend at most $x$ in the bond market for trading the stock. 
  Thus, the different ways of decomposition mean the different limits of short selling in bond, which lead to different maximal utilities and also shadow prices.
\end{remark}

\subsection{Proof of Lemma \ref{condition}}

In this subsection, we shall prove Lemma \ref{condition} by following \cite[page 814-816]{BCKMK13}. 
Thus, we only give the sketch in order to show how it develops in the num\'eraire-based context and how a positive random endowment works.  
The proof is divided in several stages. 

Firstly, similar to \cite[Lemma 4.4]{BCKMK13}, we have the following dynamic programming principle (see also \cite[Theorem 1.17]{Elk81}), which could be  proved in a direct way with the num\'eraire-based model. 
%
\begin{proposition}\label{prop311}
 Define 
   $$ \cU_s(\varphi_s^0, \varphi_s^1 ):= \esssup_{(\psi^0,\psi^1)\in\cA^\lambda_{s,T}(\varphi_s^0,\varphi_s^1)}\Econd{U(\psi^0_T+e_T)}{\cF_s}, $$
  where $\cA^\lambda_{s,T}(\varphi_s^0, \varphi_s^1)$ is the set of all admissible $\lambda$-self-financing trading strategies, 
  which agree with $\varphi\in\cA^{\lambda}(x)$ in $[0,s]$.
  
 Then, the process $\big(\cU_s(\wh \varphi_s^0,\wh \varphi_s^1)\big)_{0\leq s\leq T}$ is a martingale i.e., 
   $$ \cU_s(\wh \varphi_s^0,\wh \varphi_s^1 )= \Econd{\cU_t(\wh \varphi_t^0,\wh \varphi_t^1 )}{ \cF_s},  \quad \mbox{a.s.,} $$
  for all optimal trading strategies $\wh\varphi$ attaining $\wh{g}$.  
\end{proposition}

\begin{proof}
 Without loss of generality, it suffices to verify the following claim
  \begin{align}  \label{contra}
    \Econd{U(\psi^0_T +e_T )}{\cF_s} \leq  \Econd{U(\widehat{\varphi}^0_T +e_T )}{\cF_s},
  \end{align}
  for all $(\psi^0,\psi^1)\in\cA_{s,T}^{\lambda}(\wh \varphi_s^0,\wh \varphi_s^1)$. 
  
 To obtain a contradiction, we suppose that \eqref{contra} is not true, i.e., there exists a $(\psi^0,\psi^1)\in\cA_{s,T}^{\lambda}(\wh\varphi_s^0,\wh\varphi_s^1)$ and a set $A\subseteq\Omega$ with $\massP(A)>0$ defined as
  \begin{align}\label{setA}
    A:=\left\{\massE[U(\psi^0_T +e_T )|\cF_s] > \massE[U(\wh{\varphi}^0_T +e_T)|\cF_s]\right\} \in \cF_s.
  \end{align}
 Then define $$(\psi^0,\psi^1){\bf 1}_A +(\wh\varphi^0, \wh\varphi^1){\bf 1}_{A^c} =: (\eta^0,\eta^1)\in\cA_{s,T}^{\lambda}(\wh\varphi_s^0,\wh \varphi_s^1).$$
 We have 
   $$ \Es{U(\eta^0_T +e_T )} >  \Es{U(\wh\varphi^0_T +e_T )}=u(x;e_T), $$
  which is in contradiction to the maximality of $\widehat{\varphi}$. 
 Thus, by the definition of $\cU_s(\varphi_s^0,\varphi_s^1)$ and \eqref{contra}, we obatin
  \begin{align*}
    \cU_s(\widehat{\varphi}^0_s,\widehat{\varphi}^1_s)&=\esssup_{(\psi^0,\psi^1)\in\cA^\lambda_{s,T}(\widehat{\varphi}^0_s,\widehat{\varphi}^1_s)}\Econd{U(\psi^0_T +e_T )}{\cF_s}
       = \Econd{U(\widehat{\varphi}^0_T +e_T )}{\cF_s}\\
       &=\Econd{  \Econd{U(\widehat{\varphi}^0_T +e_T )}{\cF_t}}{\cF_s}= \Econd{\cU_t(\wh \varphi_t^0,\wh \varphi_t^1 )}{ \cF_s}.
  \end{align*} 
 Finally, the tower property of conditional expectations yields the desired result.
\end{proof}

The next step goes in an exactly same way as in \cite{BCKMK13}, i.e., we should first construct a pair $\wh Z=(\wh Z_t^0,\wh Z_t^1)_{\tT}$, then verify a shadow price can be defined by 
  $\frac{\wh Z^1}{\wh Z^0}$.
The additional positive $e_T$ in the dynamic will not alter the following results. 

\begin{proposition} \label{constructZhat}
 The following processes are well defined: 
  \begin{equation} \label{DefZtilde}
    \left\{
      \begin{array}{l} 
        \displaystyle\wt Z_t^0:= \lim_{\varepsilon\searrow 0}\frac{1}{\varepsilon}\big(\cU_t( \wh \varphi_t^0+ \varepsilon,\wh \varphi_t^1 )- \cU_t( \wh \varphi_t^0,\wh \varphi_t^1 )\big), \\
        \displaystyle\wt Z_t^1:= \lim_{\varepsilon\searrow 0}\frac{1}{\varepsilon}\big(\cU_t( \wh \varphi_t^0,\wh \varphi_t^1+ \varepsilon )- \cU_t( \wh \varphi_t^0,\wh \varphi_t^1 )\big),
      \end{array}
    \right. 
  \end{equation}
  for $0 \leq t <T$, and 
  \begin{equation} 
    \left\{\displaystyle
      \begin{array}{l}
        \wt Z_T^0 := U'(\wh \varphi_T +e_T ), \\
        \wt Z_T^1 := U'(\wh \varphi_T +e_T )(1-\la)S_T.
      \end{array}
    \right.
  \end{equation}
 Furthermore, define 
  \begin{equation} \label{DefZhat}
    \left\{
       \begin{array}{lc} 
         \displaystyle\wh Z_t^i := \lim_{\stackrel{s \searrow t}{s\in \QQ}}  \wt Z_s^i, \quad & 0 \leq t <T; \\ 
         \displaystyle\wh Z_t^i := \wt Z_T^i, \quad & t=T.
       \end{array} 
    \right. 
  \end{equation}
 Then, the process $\widehat{Z}$ is a c\`adl\`ag supermartingale and moreover, for all $0\leq t\leq T$, we have
   \begin{equation} \label{ZhatinBidAsk}
    (1-\la)S_t \leq \frac{\wh Z_t^1}{\wh Z_t^0} \leq S_t, \quad a.s. 
   \end{equation}
 Consequently, $\wh Z$ is a $\lambda$-supermartingale-CPS.
\end{proposition}

\begin{proof}[Sketch of the proof]
 Firstly, $\wt Z$ is well-defined, since the right-hand side of $\eqref{DefZtilde}$ is monotone due to the concavity of $U$  and the definition of  $\cU$. 
 Furthermore, the set 
 $$\left\{\Econd{U(\psi_T^0+e_T)}{\cF_t}\,\big|\,(\psi^0,\psi^1)\in\cA_{t,T}^{\lambda}(\wh\varphi_t+\varepsilon e_i)\right\}$$ is directed upwards, 
   thus there exists a sequence of $$(\psi^n)_{n\in\NN}=(\psi^{n,0},\psi^{n,1})_{n\in\NN} \subseteq \cA_{t,T}^{\lambda}(\wh \varphi_t+ \varepsilon e_i)$$ such that 
  \begin{align*}
    \cU_t(\wh \varphi_t+ \varepsilon e_i) &= \nearrow-\lim_{n \to \i}\Econd{U\big(\psi^{n,0}_T +e_T\big)}{\cF_t}.
  \end{align*}
 Then, by proceeding a similar procedure as in the proof of \cite[Proposition 4.5]{BCKMK13}, we can verify the supermartingale property of $\wt Z$, 
   which is not necessarily c\`adl\`ag. 
 Recalling that $u(x;e_T)$ is finitely valued and concave on $\RR_+$, we have that 
   $$ \wt y:=\widetilde{Z}^0_0=\lim_{\varepsilon\searrow 0}\frac{u(x+\varepsilon;e_T)-u(x;e_T)}{\varepsilon} $$
   takes finite values for $x>0$. 
 Consequently, by \cite[Proposition 1.3.14(iii)]{KS88}, \eqref{DefZhat} is a well-defined c\`adl\`ag supermartingale. 
 In particular, $\widehat{Z}_0^0\leq \widetilde{Z}^0_0$. 
 
 On the other hand, the proof of \eqref{ZhatinBidAsk} goes in the same way as $(57)$ in \cite{Loe00}. We complete the proof. 
\end{proof}

\begin{proof}[Proof of Lemma \ref{condition}]
 It remains to proof (2) in Lemma \ref{condition}. 
 Since $\big(\wh Z^0,\wh Z^1\big)\in\cZ^{\lambda}_{sup}$, we have
  \begin{equation} \label{xygeqXY}
    \massE\big[\wh Z^0_T\wh g\big] \leq \wh Z^0_0x. 
  \end{equation}
 
 It remains to show the reversed inequality. For $\alpha<1$, we note that $u(\alpha x;e_T)\geq \massE[U(\alpha\wh g+e_T)]$. 
 By the concavity of $u$, we obtain 
   $$ \wt Z^0_0 (x-\alpha x) \leq u(x;e_T)-u(\alpha x;e_T)\leq \massE[U(\wh g+e_T)]-\massE[U(\alpha \wh g+ e_T)]. $$ 
 Therefore, it follows from the strict concavity and the continuous differentiability of $U$ that
 \begin{equation}\label{explain} \wt Z_0^0x \leq \Es{\frac{U(\wh g+e_T) - U(\alpha\wh g +e_T)}{1-\alpha}} \leq \Es{ U'(\alpha \wh g +e_T)\wh g}.  \end{equation}
 Letting $\alpha\nearrow 1$, monotone convergence yields
  \begin{equation} \label{xyleqXY}
    \wh Z_0^0x \leq\wt Z_0^0x \leq \Es{U'(\wh g+e_T)\wh g}\leq \massE\big[\wh Z^0_T\wh g\big].    
  \end{equation}
 We complete the proof by comparing \eqref{xygeqXY} and \eqref{xyleqXY}. 
\end{proof}


\begin{remark}
 We have seen here that the nonnegativity of the random endowment is important. This ensures the strict positivity of $\alpha \wh g + e_T $, for $0< \alpha <1$, such that  $\Es{ U'(\alpha \wh g +e_T)\wh g}$ in \eqref{explain} is well defined.
\end{remark}
\begin{remark}
We  study in this paper the num\'eraire-based two-asset model just for the sake of simplicity, however, the argument could adapt to the multi-currency setting as in \cite{BCKMK13} with no essential changes. Notice that both the present paper and \cite{BCKMK13} concern a univariate utility function.
\end{remark}


\section{Example: when the random endowment becomes negative}

We have discussed the existence of shadow price processes for the utility maximization problem with positive random endowment under the no-short-selling constraint of each asset. Actually, the positivity constraint on the random endowment is a sufficient condition. However, when the random endowment becomes negative, the existence of shadow prices is not clear if we keep our setting. This will be an interesting problem for our future research. We need to mention here that with a different definition of admissibility for portfolios, the existence of sandwiched shadow price is proved in \cite{BY15}, in which the random endowment could be negative. 

In this section, we provide  a nontrivial and heuristic example showing that even when the random endowment becomes negative, there may exist a shadow price for the utility maximization problem in the market under transaction costs. 
In particular, we discuss maximal trading strategies in the context with transaction costs. 
Notice that the random endowment constructed in our example is associated with such a maximal trading strategy.



We consider the Black-Scholes model with finite time horizon. Assume that the market consists of a saving account with zero interest rate and a stock with price dynamics
  \begin{align*}
    S_t = \left\{\begin{array}{ll}
          \exp\left(B_t + \frac{t}{2}\right), & 0 \leq t \leq T/2;\\
           S_{T/2}, & T/2< t\leq T,
         \end{array} \right.
  \end{align*}
 where $(B_t)_{t\geq0}$ is a Brownian motion on a filtered probability space $(\Omega, \cF, (\cF_t)_{t\geq 0}, \massP)$. 

Now fix a level of transaction costs $\lambda\in(0,1)$. 
We assume that an agent is endowed with the initial capital $x=1$ at time $0$ and the random endowment $e_T:=-\zeta(1-\lambda)S_{T/2}$ at terminal time $T$, 
  where $\zeta$ is an $\mathcal{F}_{T}$-measurable $(0,1)$-valued uniformly distributed random variable which is independent of $\mathcal{F}_{T/2}$. 
Then, we have the following proposition.
  \begin{proposition}\label{primopt}
  The buy-hold-sell strategy 
  \begin{equation}\label{strategy}
    \big(\widehat{\varphi}^0_t,\widehat{\varphi}^1_t\big):=
     \left\{
       \begin{array}{ll}
         (1,0), & t=0;   \\
         (0,1), & 0<t<T/2;  \\
         \big((1-\lambda)S_{T/2},0\big), & T/2\leq t\leq T,
       \end{array}
     \right. 
  \end{equation}
  solves the utility maximization problem (\ref{pp}).
  \end{proposition}  
  
Before proceeding the proof of this proposition, we first introduce the notion of maximal trading strategy in this market with transaction costs. 
Notice that a similar notion in frictionless markets is created by Delbaen and Schachermayer in \cite{DS94} for considering a superreplication problem.

\begin{definition}
An element $\varphi^0_T\in \mathcal{C}^\lambda(1)$ is called maximal, if for $\phi^0_T\in \mathcal{C}^\lambda(1)$ satisfying $\phi^0_T\geq \varphi^0_T$, a.s., we have $\phi^0_T= \varphi^0_T$, a.s. 
A trading strategy $(\varphi^0, \varphi^1)$ is called maximal in $\mathcal{A}^\lambda(1)$, 
if it is associated with $(\varphi^0_T, 0)$, where $\varphi^0_T$ is  a maximal element in $\mathcal{C}^\lambda(1)$.
\end{definition}

\begin{remark}
The existence of a maximal element in $\mathcal{C}^\lambda(1)$ could be deduced by transfinite induction, provided that $\mathcal{C}^\lambda(1)$ is bounded in $L^0$ and closed w.r.t.~convergence in probability (cf.~Schachermayer \cite{Sch14super}).
\end{remark}

Indeed, the trading strategy defined by \eqref{strategy} is maximal in $\mathcal{A}^\lambda(1)$. 

\begin{proposition}\label{maxima}
The random variable $(1-\lambda)S_{T/2}$ is a maximal element in $\mathcal{C}^\lambda(1)$. Therefore, the trading strategy $(\widehat{\varphi}^0, \widehat{\varphi}^1)$ defined by \eqref{strategy} is maximal in $\mathcal{A}^\lambda(1)$.  
\end{proposition}

Before proving the proposition above, we first give the proof of Proposition \ref{primopt} with the help of the maximality of  (\ref{strategy}).

\begin{proof}[Proof of Proposition \ref{primopt}]
Consider an element $\phi^0_T\in \mathcal{C}^\lambda(1)$ such that $\phi^0_T\neq \widehat{\varphi}^0_T= (1-\lambda)S_{T/2}$. From the maximality of $(\widehat{\varphi}^0, \widehat{\varphi}^1)$, we have $\massP\big[\phi^0_T<\widehat{\varphi}^0_T\big]>0$. Therefore, there exists $k\in\NN$, such that
   $$ \massP\left[\phi^0_T< \big(1- \tfrac{1}{k}\big)(1-\lambda)S_{T/2}\right] >0. $$
Note that $\left\{\phi^0_T< \big(1- \tfrac{1}{k}\big)(1-\lambda)S_{T/2}\right\}\in \mathcal{F}_{T/2}$ and $\zeta$ is independent of $\mathcal{F}_{T/2}$. Then, 
   $$ \massP\left[\phi^0_T< \big(1- \tfrac{1}{k}\big)(1-\lambda)S_{T/2} \mbox{ and } 1-\zeta <\tfrac{1}{k}\right]>0, $$
which implies $\massP[\phi^0_T+e_T<0] >0$. Consequently, $\massE[U(\phi^0_T+e_T)]=-\infty$. 

On the other hand, by the definition of $\zeta$, we have 
 $$ \widehat{\varphi}^0_T+e_T=(1-\zeta)(1-\lambda)S_{T/2}>0,\ \mbox{a.s.} $$
Thus, $u(1; e_T)=\massE[U(\widehat{\varphi}^0_T+e_T)]>-\infty$. The proof is completed. 
\end{proof}

\begin{remark}
 Even without the no-short-selling constraint, $(\widehat{\varphi}^0, \widehat{\varphi}^1)$ is still the unique strategy on $[0, T/2]$ solving the problem (\ref{pp}). 
 Indeed, we shall prove in the next subsection that $(\widehat{\varphi}^0, \widehat{\varphi}^1)$ is a maximal trading strategy in a larger space introduced in \cite{CS15duality, CSY15} (still denoted by $\mathcal{A}^\lambda(1)$).
\end{remark}

Next, we shall construct a shadow price in the following corollary. 
It is evident that the shadow price is not unique in our case. 
\begin{corollary}
 Define 
 \begin{align*}
    \widetilde S_t := \left\{\begin{array}{ll}
            S_0,  &t=0;\\
          \exp\left(B_t + \left(\frac{1}{2}+\frac{2\log(1-\lambda)}{T}\right)t\right), & 0 < t < T/2;\\
           (1 - \lambda)S_{T/2}, & T/2\leq  t\leq T.
         \end{array} \right.
  \end{align*}
Then, the process $\widetilde{S}$ is a shadow price of the utility maximization problem (\ref{pp}).
\end{corollary}

\begin{proof}
 It follows from the definition of $\widetilde{S}$ that $\widetilde{S}\in[(1-\lambda)S,S]$. 

 Obviously, the wealth process associated with the buy-hold-sell strategy under $\widetilde{S}$ is $\widetilde{S}$ itself. 
 Thus, we can find an equivalent martingale measure $\massQ\in \mathcal{M}^e(\widetilde{S})$, 
   under which this wealth process is a uniformly integrable martingale. By \cite[Corollary 4.6]{DS95no}, $\widehat{\varphi}^0_T=\widetilde{S}_T$ is a maximal element in $\mathcal{C}^{\widetilde{S}}(1)$, 
   where 
    $$ \mathcal{C}^{\widetilde{S}}(1):=\left\{\widetilde{X}_T=1+\big(H\sint\widetilde{S}\big)_T\,\Big|\,H\mbox{ is admissible}\right\}. $$
 Similarly to the proof of Proposition \ref{primopt}, one can see that $\widehat{\varphi}^0_T=\widetilde{S}_T$ solves the utility maximization problem in the frictionless market $u^{\widetilde{S}}(1; e_T)$, and 
  \begin{align*}
     u^{\widetilde{S}}(1; e_T) &= \sup_{\widetilde{X}_T\in \mathcal{C}^{\widetilde{S}}(1)}\Es{U\big(\widetilde{X}_T+e_T\big)}
                    = \Es{U\big(\widetilde S_T - \zeta(1-\lambda){\widetilde S_T}\big)} \\
                   &= \Es{U\big((1- \zeta)(1- \lambda)S_T\big)} = u(1; e_T),
  \end{align*}
  which implies that $\widetilde{S}$ is the desired shadow price. 
\end{proof}

We now prove that (\ref{strategy}) defines a maximal trading strategy in $\mathcal{A}^\lambda(1)$, i.e., Proposition \ref{maxima}. Here, $\mathcal{A}^\lambda(1)$ could be the set of admissible trading strategies defined in \cite{CS15duality} without no-short-selling constraint.
\begin{proof}[Proof of Proposition \ref{maxima}]
 Let $\varphi=(\varphi_t^0, \varphi_t^1)_\tT$ be another trading strategy in $\mathcal{A}^\lambda(1)$. 
 Since the stock-price process $S$ is continuous and the predictable trading strategies are pathwisely of finite variation, we could only consider those 
 $(\varphi^0, \varphi^1)$ 
 with right-continuous paths (cf.~\cite{CSY15}). In addition, we assume that the agent buys no more stocks after time $T/2$. 
 
 In what follows, we shall discuss trading strategies $\varphi$ in two cases and furthermore, we shall prove that the following statement cannot hold
 \begin{equation}\label{stat}
 \varphi^0_T\geq \widehat{\varphi}^0_T=(1-\lambda)S_{T/2},\ \mbox{a.s.\ with\ } \massP[\varphi^0_T>\widehat{\varphi}^0_T]>0,
 \end{equation}
 and thus $\widehat{\varphi}$ is a maximal trading strategy in $\mathcal{A}^\lambda(1)$.
 
 \vspace{3mm}
 
\noindent {\bf{Case I.}} Suppose that there exists a $\gamma>0$ and a set $B\in \mathcal{F}_{T/2}$ with $\massP(B)>0$ such that for $\omega\in B$, $\varphi^1_t(\omega)\leq 1-\gamma$, $0\leq t\leq T/2$,  $\massP$-a.s. (i.e., the holding in stock never excesses $1-\gamma$). Our aim is to construct a set $\overline{B}\in\mathcal{F}_{T/2}$ with $\massP(\overline{B})>0$, such that for $\omega\in \overline{B}$, $\varphi^0_T(\omega)<\widehat{\varphi}^0_T(\omega)$. 

 For $\beta>0$, define $\tau_\beta:=\inf\{t>0: S_t\geq 1+\beta\}\wedge T/2$. Then, by choosing $\beta\ll\lambda$ sufficiently small, we can find an $\mathcal{F_{\tau_\beta}}$-measurable set defined by 
%
$$
B':=\{\omega: \tau_\beta(\omega)<T/2\}\cap\{\omega: 1-\lambda/3\leq S_t(\omega)\leq 1+\beta,\ 0\leq t\leq \tau_\beta(\omega)\},
$$
such that  $\massP(B')>1-\massP(B)$.
It is obvious that $\massP(B' \cap B)>0$ and for $\omega\in B'$, $S_{\tau_\beta(\omega)}(\omega)=1+\beta$. Furthermore, we define
$$
\widetilde{B}':=B'\cap\{\omega: \varphi^1_t(\omega)\leq 1-\gamma,\ 0\leq t\leq \tau_{\beta}(\omega)\}\supseteq B'\cap B.
$$
From $\massP(B'\cap B)>0$, we have $\massP(\widetilde{B}')>0$.
   \begin{figure}
   \centering
    \includegraphics[width=12cm]{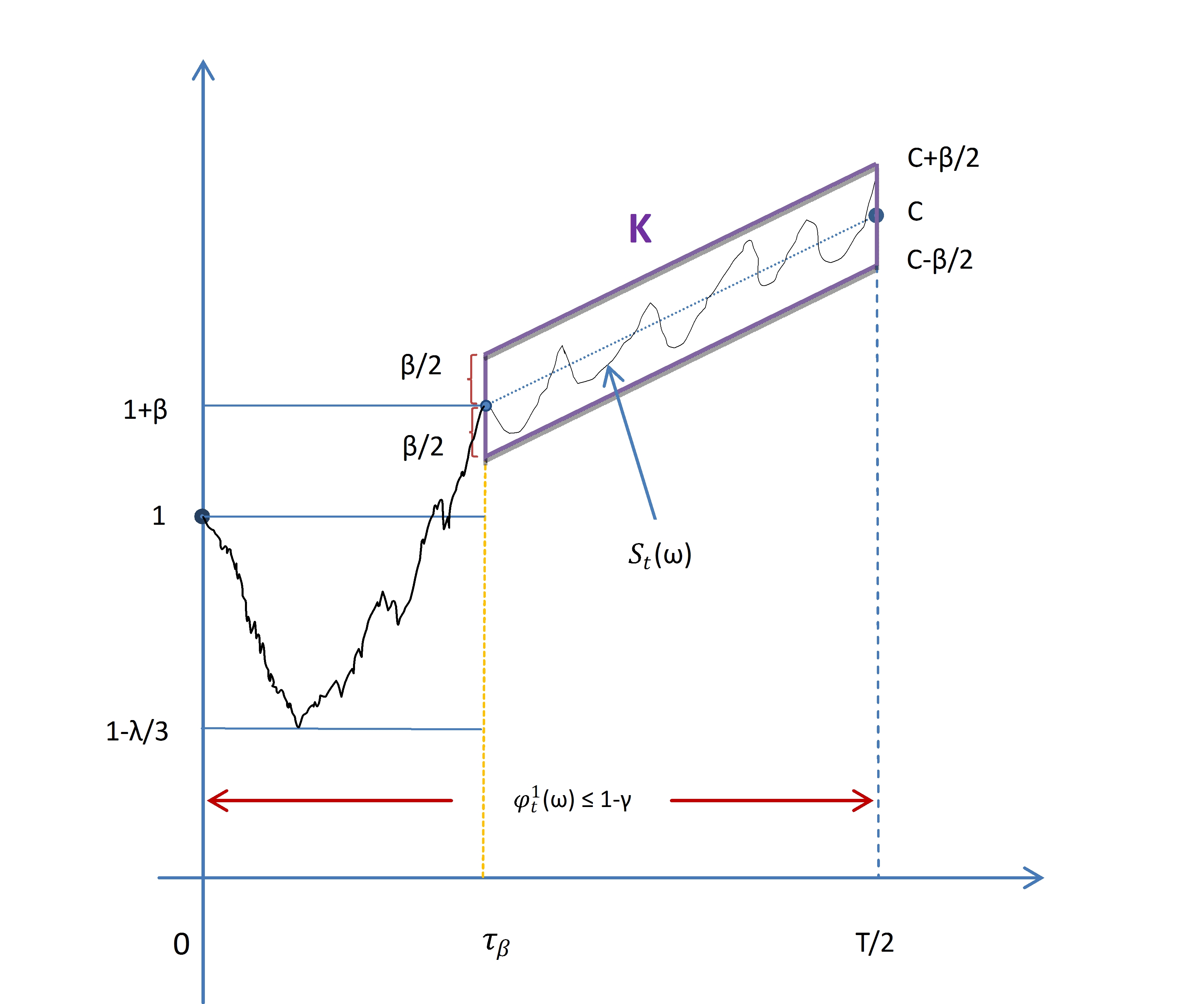}\caption{Case 1}\label{Figure1}
  \end{figure} 
We  now construct a corridor $K$ of stock price on the time interval $\llbracket \tau_\beta,T/2\rrbracket$ with width $\beta$ (see Figure \ref{Figure1}). The center point of the corridor at $\tau_\beta$ is $1+\beta$ while at $T/2$ is a sufficiently large number $C$. Obviously, the stock-price process $S$ has conditional full support (see, e.g., \cite{GRS08}), which implies the set  
  $$\overline{B}':=\widetilde{B}'\cap\{S\mbox{ stays in the corridor }K\ \mbox{on}\ \llbracket \tau_\beta, T/2\rrbracket\big\},$$
  has a strictly positive measure, i.e., $\massP(\overline{B}')>0$. 
For $\omega\in \overline{B}'$, the most advisable trading strategy is to buy $1-\gamma$ stock shares at the lowest price, which should be at least $1-\lambda/3$, and hold them until $T/2$. Therefore,
    $$\varphi^0_T(\omega)=V^{liq}_{T/2}(\varphi)(\omega)<(1-\lambda)\big(C + \tfrac{\beta}{2}\big)(1- \gamma)+1.$$
On the other hand,
    $$\widehat{\varphi}^0_{T}(\omega)=\widehat{\varphi}^0_{T/2}(\omega)>(1-\lambda)\big(C-\tfrac{\beta}{2}\big).$$ 
  By properly
  choosing {$C > \frac{1}{\gamma(1-\lambda)} + \frac{\beta(2-\gamma)}{2\gamma}$}, we could ensure that $\varphi^0_T(\omega)<\widehat{\varphi}^0_T(\omega)$.

\vspace{3mm}

{\noindent\bf{Case II.}} Suppose that Case I cannot happen, then for almost every $\omega\in \Omega$, the following stopping time takes values 
in $[0, T/2]$:
$$\sigma:=\inf\{t>0:\varphi^1_t\geq 1\ \mbox{or}\ \varphi^1_{t-}\geq 1\}.$$
For any stopping times $\tau$ and $\rho$ taking values in $[0, T/2]$, $\rho\geq \tau$, 
we have:
   \begin{align}\label{pq1}
  V^{liq}_{\rho}(\varphi) 
                &= \varphi_\tau^0 + \varphi_\tau^1 S_\tau + \int_\tau^\rho \varphi_t^1 dS_t 
                    - \lambda \int_\tau^\rho S_t d\varphi_t^{1,\downarrow}-\lambda S_\rho (\varphi_\rho^1)^+\notag\\
                  &=  \varphi_\tau^0 + \varphi_\tau^1 S_\tau + \int_\tau^\rho \varphi_t^1 dS_t
                    - \lambda \int_\tau^\rho S_t d\varphi_t^{1,\downarrow}-\lambda S_\rho (\varphi_\rho^1)^+ 
                     + \lambda \int_\tau^\rho \varphi_t^1 dS_t  - \lambda \int_\tau^\rho \varphi_t^1 dS_t\notag\\
                  &=  \varphi_\tau^0 + (1 - \lambda)\varphi_\tau^1 S_\tau  +  (1 - \lambda) \int_\tau^\rho \varphi_t^1 dS_t
                   -\lambda S_\rho (\varphi_\rho^1)^- - \lambda \int_\tau^\rho S_t d\varphi_t^{1,\uparrow}.
   \end{align}
%
%
%
In the sequel, we shall consider two situations, i.e., 
\begin{enumerate}[(i)]
\item For almost every $\omega\in \Omega$, $$\varphi_\sigma^0 (\omega) + (1 - \lambda)\varphi_\sigma^1(\omega)  S_\sigma(\omega) \geq \widehat \varphi_\sigma^0(\omega)  + (1 - \lambda)\widehat \varphi_\sigma^1(\omega)  S_\sigma(\omega);$$
\item There exists an $\mathcal{F}_\sigma$-measurable set $D$, with $\massP(D)>0$, for $\omega\in D$, $$\varphi_\sigma^0 (\omega)+(1-\lambda)\varphi_\sigma^1(\omega)  S_\sigma(\omega) <\widehat \varphi_\sigma^0(\omega) +(1-\lambda)\widehat \varphi_\sigma^1(\omega)  S_\sigma(\omega).$$
\end{enumerate}
Assume that (i) holds, then we observe that for almost every $\omega\in \Omega$, 
\begin{equation}\label{21}\varphi_\sigma^0 (\omega) + (1 - \lambda)\varphi_\sigma^1(\omega)  S_\sigma(\omega) = \widehat \varphi_\sigma^0(\omega)  + (1 - \lambda)\widehat \varphi_\sigma^1(\omega)  S_\sigma(\omega).
\end{equation}
Indeed, if $\varphi^1_{\sigma(\omega)} (\omega)\geq 1$, then 
  $$ \varphi_\sigma^0 (\omega) + (1 - \lambda)\varphi_\sigma^1(\omega)  S_\sigma(\omega)=V^{liq}_{\sigma(\omega)}(\varphi)(\omega)\leq V^{liq}_{\sigma(\omega)-}(\varphi)(\omega); $$
 if $\varphi^1_{\sigma(\omega)} (\omega)< 1$, then $\varphi^1_{\sigma(\omega)-} (\omega)\geq 1$, and thus from the self-financing constraint, 
  $$ \varphi_\sigma^0 (\omega) + (1 - \lambda)\varphi_\sigma^1(\omega)  S_\sigma(\omega)\leq \varphi_{\sigma-}^0 (\omega) + (1 - \lambda)\varphi_{\sigma-}^1(\omega)  S_{\sigma}(\omega)=V^{liq}_{\sigma(\omega)-}(\varphi)(\omega). $$ 
Consider a frictionless market in which the agent trades for $(1-\lambda)S$, where the agent gains better than trading in the  market with transaction costs. Then, we can deduce
\begin{align*}
\widehat\varphi_0^0 + \widehat\varphi_0^1 S_0+(1-\lambda)\int^\sigma_0 \widehat\varphi^1_tdS_t&= \widehat \varphi_\sigma^0  + (1-\lambda)\widehat \varphi_\sigma^1S_\sigma 
   \leq \varphi_\sigma^0 + (1 - \lambda)\varphi_\sigma^1  S_\sigma\\
  &\leq V^{liq}_{\sigma-}(\varphi)\leq \varphi_0^0 + \varphi_0^1 S_0+(1-\lambda)\int^\sigma_0 \varphi^1_tdS_t.
\end{align*}
Since $\varphi_0^0 + \varphi_0^1 (1-\lambda)S_0=\widehat \varphi_0^0 +\widehat \varphi_0^1(1-\lambda) S_0=1$ and ${\bf 1}$ is a maximal trading strategy when trading for $(1-\lambda)S$, we can prove (\ref{21}).

Now we calculate $V^{liq}_{T}(\varphi)$. From (\ref{pq1}), we obtain
\begin{equation} \label{ea3}
\begin{aligned}
 V^{liq}_{T}(\varphi) &\leq  {\varphi}^0_\sigma + (1-\lambda)\varphi_\sigma^1 S_\sigma +(1 - \lambda)\int^T_\sigma \varphi^1_t dS_t \\
 & =\widehat{\varphi}_0^0 + (1-\lambda)\wh\varphi_0^1 S_0 + (1-\lambda)\int^\sigma_0 \widehat \varphi^1_t dS_t+(1 - \lambda) \int_\sigma^T  \varphi_t^1 dS_t \\
&  =\widehat{\varphi}_0^0 + (1-\lambda)\wh\varphi_0^1 S_0 + (1-\lambda)\int^T_0 \left(\widehat \varphi^1_t{\bf 1}_{\llbracket 0, \sigma\rrbracket}(t)+
\varphi^1_t{\bf 1}_{\llbracket \sigma, T\rrbracket}(t)\right)
 dS_t.
\end{aligned}
\end{equation}
It is obvious that trading for $(1-\lambda)S$, $((\widehat \varphi^1_\cdot{\bf 1}_{\llbracket 0, \sigma\rrbracket}(\cdot)+
\varphi^1_\cdot{\bf 1}_{\llbracket \sigma, T\rrbracket}(\cdot))\sint (1-\lambda)S)$ is uniformly bounded from below, and thus admissible. Comparing (\ref{ea3}) with
   \begin{align*}
V^{liq}_{T}(\widehat\varphi) &= \widehat{\varphi}_0^0 + (1-\lambda)\wh\varphi_0^1 S_0 + (1-\lambda)\int_0^T \widehat \varphi_t^1 dS_t,
   \end{align*}
we can conclude from the maximality of $(1-\lambda)\int^T_0{\bf 1}dS_t$ that \eqref{stat} cannot happen.

On the other hand, if (ii) holds, then by similar reasoning, we have for $\omega\in D$,
\begin{equation} \label{ere}
\begin{aligned}
 V^{liq}_{T}(\varphi)
 & <\widehat{\varphi}_0^0 + (1-\lambda)\wh\varphi_0^1 S_0 + (1-\lambda)\int^\sigma_0 \widehat \varphi^1_t dS_t+ (1 - \lambda)\int_\sigma^T  \varphi_t^1 dS_t \\
&  =\widehat{\varphi}_0^0 + (1-\lambda)\wh\varphi_0^1 S_0 + (1-\lambda)\int^T_0 \left(\widehat \varphi^1_t{\bf 1}_{\llbracket 0, \sigma\rrbracket}(t)+
\varphi^1_t{\bf 1}_D{\bf 1}_{\llbracket \sigma, T\rrbracket}(t)\right)
 dS_t.
\end{aligned} 
\end{equation}

Again, $((\widehat \varphi^1{\bf 1}_{\llbracket 0, \sigma\rrbracket}+\varphi^1{\bf 1}_D{\bf 1}_{\llbracket \sigma, T\rrbracket})\sint (1-\lambda)S)$ is uniformly bounded from below when trading for $(1-\lambda)S$, and thus admissible. 
Comparing $\big((\widehat \varphi^1{\bf 1}_{\llbracket 0, \sigma\rrbracket}+\varphi^1{\bf 1}_D{\bf 1}_{\llbracket \sigma, T\rrbracket})\sint (1-\lambda)S\big)_T$ with $({\bf 1}\sint (1-\lambda)S)_T$, 
  we can conclude that either for almost every $\omega\in D$, 
  $$ \big((\widehat \varphi^1{\bf 1}_{\llbracket 0, \sigma\rrbracket}+\varphi^1{\bf 1}_D{\bf 1}_{\llbracket \sigma, T\rrbracket})\sint (1-\lambda)S\big)_T= ({\bf 1}\sint (1-\lambda)S)_T, $$
or there exists a subset $D'\subset D$, $D'\in \mathcal{F}_{T/2}$,
such that for $\omega\in D'$, 
 $$ \big((\widehat \varphi^1{\bf 1}_{\llbracket 0, \sigma\rrbracket}+\varphi^1{\bf 1}_{D}{\bf 1}_{\llbracket \sigma, T\rrbracket})\sint (1-\lambda)S\big)_T< ({\bf 1}\sint (1-\lambda)S)_T. $$
In particular, we exclude \eqref{stat} by recalling (\ref{ere}).
\end{proof}

\section{Appendix}

In the appendix, we prove the technical lemma, which was used to prove our main result in Section 3.
%
%
%
%
\begin{lemma}\label{lemnew}
 Suppose $\{f_n\}_{n\in \mathbb N}$ is a sequence in $L^1$, $f_n\rightarrow f_0\in L^0$, almost surely. 
Moreover, $\lim_{n\rightarrow \infty}\massE[f_n]$ is finite and $f^+_0$ is integrable. We denote $\alpha:= \lim_{n\rightarrow \infty}\massE[f_n]-\massE[f_0]$. In particular, if $\massE[f_0]=-\infty$, we note $\alpha:=\infty$.
  Then, we have for any $M>0$, 
            \begin{align}\label{poi}
              \limsup_{n\rightarrow \infty} \massE\big[f_n{\bf 1}_{\{f_n\geq M\}}\big]\geq \alpha.
            \end{align}
\end{lemma}
 
\begin{proof}
We suppose, contrary to our claim \eqref{poi}, there exists $M>0$, such that 
     \begin{align} \label{poe}
       \limsup_{n\rightarrow \infty} \massE\big[f_n{\bf 1}_{\{f_n\geq M\}}\big]=:\beta< \alpha,
     \end{align}
     which implies the boundedness of $\{\massE[f_n{\bf 1}_{\{f_n\geq M\}}]\}_{n\in\mathbb{N}}$.
    Suppose $\massE[f_0]=-\infty$, then we have
      $$ \massE\big[f_0{\bf 1}_{\{f_0< M\}}\big]\leq \massE[f_0]=-\infty. $$ 
  Thus,
    \begin{align} \label{loe}
      \limsup_{n\to\infty}\massE\big[f_n {\bf 1}_{\{f_n< M\}}\big]\leq \massE\big[f_0{\bf 1}_{\{f_0< M\}}\big]=-\infty.
    \end{align}
  From \eqref{loe} and \eqref{poe}, we can conclude that $\lim_{n\rightarrow \infty}\{\massE[f_n]\}_{n\in \mathbb{N}}=-\infty$, which contradicts to the assumption.
  Therefore, $\massE[f_0]\in\RR$. 
  In this case, we have 
   \begin{equation} \label{Ef01>Ef0}
    \begin{aligned} 
     \massE\big[{f_0{\bf 1}_{\{f_0<M\}}}\big] &\geq \limsup_{n\rightarrow \infty}\massE\big[{f_n{\bf 1}_{\{f_n<M\}}}\big] \\
           &= \lim_{n\rightarrow \infty}\massE[{f_n}]-\liminf_{n\rightarrow \infty}\massE\big[{f_n{\bf 1}_{\{f_n\geq M\}}}\big] \\
           &\geq \massE[f_0]+\alpha-\beta,
    \end{aligned}
   \end{equation}
   where the equality is deduced from the convergence of $\{\massE[f_n]\}_{n\in \mathbb{N}}$ and the boundedness of $\{\massE[f_n{\bf 1}_{\{f_n\geq M\}}]\}_{n\in\mathbb{N}}$.
  Obviously, \eqref{Ef01>Ef0} is a contradiction.
  \end{proof}
  \begin{remark}
 A generalized version of this lemma also holds, which is similar to \cite[Lemma 3.16]{Sch04LN} (ii) (see Gu \cite{Gu17} for the detailed proof). Precisely, if we assume the same as the above lemma, then for any $\alpha'<\alpha$ and $M>0$, there exists a subsequence $\{f_{n_k}\}_{k\in \mathbb{N}}$ and a sequence of disjoint sets $(A_k)_{k\in\mathbb{N}}$, 
            such that for each $k\in \mathbb{N}$, $f_{n_k}\geq M$ on $A_k$, and
             \begin{align*}\massE\big[f_{n_k}{\bf 1}_{A_k}\big]\geq \alpha'. \end{align*}
  \end{remark}

\noindent {\bf Acknowledgement:} The authors gratefully acknowledge financial support from the Austrian Science Fund (FWF) under grant P25815 and from the European Research Council under ERC Advanced Grant 321111. This work is partially done during the visit of L. Gu and J. Yang hosted by Prof. N. Touzi at CMAP, \'Ecole Polytechnique, which is very much appreciated. The authors are grateful to Prof.~W.~Schachermayer for his suggestions on the Appendix and to the anonymous reviewers for their kind comments.
\bibliography{SPwithRE_Rev3}
\bibliographystyle{abbrv}
 
\end{document}